\PassOptionsToPackage{sort&compress}{natbib}
\documentclass[preprint,11pt,authoryear]{elsarticle}

\usepackage[margin=1in]{geometry}
\usepackage{amsmath,amsfonts,amssymb,amsthm}
\usepackage{mathtools}
\usepackage[dvipsnames]{xcolor}
\usepackage{derivative}
\usepackage{mathrsfs}
\usepackage{booktabs}
\usepackage{float}
\usepackage{subcaption}
\usepackage{hyperref}

\usepackage[twoside]{fancyhdr}
\renewcommand{\vec}[1]{{\mathbf{#1}}}

\newcommand{\R}{\mathbb{R}}

\newcommand{\Ascr}[0]{\mathscr{A}}
\newcommand{\Bscr}[0]{\mathscr{B}}
\newcommand{\Cscr}[0]{\mathscr{C}}

\renewcommand{\Pr}{\mathbb{P}}
\newcommand{\E}[0]{\mathbb{E}}
\newcommand{\normal}[0]{\mathcal{N}}

\newcommand{\Zero}[0]{\mathbf{0}}

\newcommand{\mbf}[1]{{\mathbf{#1}}}
\newcommand{\bs}[1]{{\boldsymbol{#1}}}

\newcommand{\tr}[0]{\operatorname{tr}}
\newcommand{\diag}[0]{\operatorname{diag}}

\newtheorem{theorem}{Theorem}[section]

\theoremstyle{definition}
\newtheorem{definition}{Definition}
\newtheorem{setting}[theorem]{Setting}
\newtheorem{assumption}[theorem]{Assumption}

\theoremstyle{remark}

\newcommand{\Var}[0]{\operatorname{Var}}

\newcommand{\CFVaR}[0]{{\operatorname{CFVaR}}}

\newcommand{\Lc}{\mathcal{L}}

\newcommand{\Acal}{\mathcal{A}}
\newcommand{\Bcal}{\mathcal{B}}
\newcommand{\Ccal}{\mathcal{C}}

\newcommand{\Fcal}{\mathcal{F}}

\hypersetup{
    colorlinks={True},
    linkcolor={blue},
    citecolor={SpringGreen},
}

\renewcommand{\Pr}{\mathbb{P}}

\DeclareMathOperator{\Sharpe}{Sharpe}
\newcommand{\AdjSharpe}[0]{{\operatorname{R-VaR}_2^\alpha}}

\begin{document}
\begin{frontmatter}

\title{Sharpe Ratio and Return-VaR Ratio Maximization for Option Portfolios with Skew-Elliptical $t$ Underlying Returns}

\author[mcmaster]{Kyle Sung} \corref{cor1}
\author[mcmaster]{Traian A. Pirvu}
\address[mcmaster]{Department of Mathematics and Statistics, McMaster University}
\date{\today}

\cortext[cor1]{Correspondence to: Kyle Sung \textless sungk5@mcmaster.ca\textgreater}

\begin{abstract}
    We provide a formulation for optimal option portfolios under Sharpe Ratio maximization when the underlying returns follow a skew-elliptical $t$-distribution. This departs from the traditional normal returns setting in the context of Sharpe ratio maximization by allowing the modelling of heavy-tailed and skewed dynamics. The novelty of this paper and our main result is to provide explicit formulas for the portfolio weights when maximizing the Sharpe ratio and return-to-Value-at-Risk (VaR) ratio in the skew-elliptical setting. Numerical experiments reveal that the optimal portfolios for the two ratios are different.
\end{abstract}
\end{frontmatter}


\section{Introduction}

The Sharpe ratio \citep{sharpe1998sharpe}, defined as the ratio between the excess return a portfolio experiences above the market and the volatility of the portfolio, is one of the most commonly used portfolio performance metrics \citep{maller2010}, operating under the principle of finding a trade off between maximizing reward and minimizing risk. 
Indeed, it has become one of the most widespread measures of portfolio performance \citep{10.1093/rfs/hhm075};
researchers have explored findings  maximal Sharpe ratio allocations for stock portfolios \citep{risks5020027}. 

The portfolio optimization problem becomes more difficult when the portfolio consists of options instead of simply the underlyings. Options are derivatives that give the purchaser of the option the right, but not the obligation, to buy or sell an underlying asset at a pre-determined price and date. The nonlinearity of the payoff of an option means that the distribution of the portfolio return is often intractable. Thus, the delta-gamma approximation \citep{Glasserman2003} is often used. \citet{CUI20132124,Chen2013} develop option portfolio optimization for value-at-risk-based optimization using the Cornish-Fisher quantile approximation in the case of normal returns. 

The financial markets have been found to display behaviours that are not well described by a normal distribution \citep{Glasserman2003}, exhibiting both heavy tails and asymmetry \citep{Cotter01102006,Echaust20142234}. As such, much research has considered incorporating heavier tails \citep{basnarkov2019optionpricingheavytaileddistributions,Cassidy_2010,Cassidy01082013} and skewness \citep{AzzaliniCapitanio2003Skewt,YinBalakrishnan2024SkewElliptical,AzzaliniCapitanio1999SkewNormal,Wang2024SkewT,KIM2003417,genton2005generalized,GENTON2001319,BRANCO200199,Lin2007RobustSkewT}.
\citet{Hu01012010} consider portfolio optimizations for stock returns under Student $t$ and skew $t$-distributions, but in the equity space instead of the option space. 

Most recently, work by \citet{sung2026optimaloptionportfoliosskewelliptical} has bridged these two literatures, solving optimal optimal portfolios for variance and Value at Risk (VaR) minimization when the underlying returns follow a skew-elliptical $t$-distribution; they obtain closed-form tractable weights in both cases.

Though the literature has considered risk minimization for option portfolios in the skew $t$ setting \citep{sung2026optimaloptionportfoliosskewelliptical}, to the best of our knowledge, a treatment of maximization of a risk-reward ratio is not yet known. 

The Sharpe ratio measures the additional return per unit of volatility and has become a standard tool for evaluating portfolio performance. However, its reliance on the standard deviation as the sole risk metric may misrepresent risk when returns deviate from normality, prompting adaptations.

This study extends previous portfolio optimization research by considering the Sharpe ratio and an adjusted Sharpe ratio, both in the context of option portfolios where the underlying asset returns follow a skew $t$-distribution. 

\subsection{Related Work} 

\paragraph{Risk-Reward Ratios} Many have tried to improve on the Sharpe ratio. The Sortino ratio \citep{sortino1991} replaces the variance with the downside deviation of the portfolio; the STARR ratio \citep{MartinR.Douglas2003Popa} replaces the variance with the expected shortfall. \citet{ZAKAMOULINE20091242} propose a family of generalized Sharpe ratios which consider other measures than the mean and variance. \citet{Cogneau2009_101Ways} propose 101 ways to measure portfolio performance. The return-to-VaR ratio has also be proposed \citep{AlexanderGordonJ.2003PPEU,GregoriouGregN2004PoCh} for downside loss control. 
Omega ratio was considered by \citet{risks5020027} in portfolio optimization. Risk-reward option portfolios have also been considered in the normal setting \citep{Jewell2013}. 

\paragraph{Heavy-Tailed Portfolio Optimization} In the objective of risk minimization, various approaches to portfolio optimization have been considered. \citet{Hu01012010} numerically consider the optimal risk minimization of portfolios when the underlying returns are Student $t$ and skewed $t$-distributed. \citet{sung2026optimaloptionportfoliosskewelliptical} analytically consider risk minimization of option portfolios in the skew-elliptical $t$-distribution setting. Research has also considered maximum Sharpe ratios under other distributions.

To the best of our knowledge, the general question concerning Sharpe ratio maximization when the underlying returns follow a skew-elliptical $t$-distribution, is unknown.

\paragraph{Contributions and Findings} Our contribution is to provide closed-form expressions for the portfolios which maximize the Sharpe ratio and R-VaR ratio. The departure from the traditional normal returns setting to the skew-elliptical $t$ setting allows investors to incorporate both heavy tails and skewness into their optimization. This makes the model more realistic and applicable to real-world capital markets.

Our findings are twofold: first, we demonstrate that our analytic portfolio optimization results yield optimal Sharpe and R-VaR portfolios with different compositions. However, these objectives tend to yield highly leveraged positions. We address this by imposing portfolio box constraints and performing numerical portfolio optimization.

\subsection{Outline}

The remainder of the paper is organized as follows: section~\ref{sec:model} provides an overview of the setting, assumptions, and background needed for heavy-tailed skewed portfolio optimization.
Section~\ref{sec:portfolio_optimization} presents our main results: optimal options portfolios under Sharpe ratio and R-VaR ratio maximization assuming returns follow a skew $t$-distribution. Section~\ref{sec:params} introduces the main dataset for numerical experiments, which give examples of optimal portfolios in 
Section~\ref{sec:experiments}. Section~\ref{sec:conclusion} concludes. All major proof details are relegated to~\ref{sec:proofs}.  

\section{The Model} \label{sec:model}

\subsection{Notation}

All random variables are henceforth defined on the probability space $(\Omega, \Fcal, \Pr)$, and expectations are taken with respect to the probability measure $\Pr$. 
Let $\odot$ denote the Hadamard product. 
For a matrix $X \in \R^{M\times N}$, we let $X_{[\bullet, i]}$ denote its $i^{\text{th}}$ column.

\subsection{Setting}

\noindent We will work in the setting of \citet{sung2026optimaloptionportfoliosskewelliptical}. For convenience, we restate the main components below.

\begin{setting}[Optimal Portfolios]
    \label{setting:optimal_portfolios}
    Consider a portfolio consisting of $M$ options written on $N$ underlying stocks. Let $\vec{x} \coloneqq (x_1,\ldots,x_M)$ denote the vector of number of shares held of each option. Denote by 
    \begin{equation}
        \label{eqn:option_prices_vector}
        \vec{v}(\vec{S},t) \coloneqq \big(V_1(\vec{S},t),\ldots,V_M(\vec{S},t)\big)
    \end{equation} 
    the option values at time $t$ with underlying prices $\vec{S} \coloneqq (S_1,\ldots,S_N)$. The portfolio value is given by 
    \[
        V(\vec{x};\vec{S},t) 
        = 
        \vec{x}^\top \vec{v}(\vec{S},t) 
        .
    \]
    Note that we can easily convert from shares to weights as follows:
    \[
    \vec{w} = \dfrac{\vec{v} \odot \vec{x}}{\vec{v}^\top \vec{x}}
    .
    \]
    Let $\Delta V (\vec{x})$ denote the portfolio P\&L over the time interval $[t,t + \Delta t]$, which can be written explicitly as \begin{equation}
        \Delta V(\vec{x}) = V(\vec{x}; \vec{S} + \Delta \vec{S}, t + \Delta t) - V(\vec{x}; \vec{S},t)
        .
    \end{equation}
\end{setting}

\noindent Let $r_f^{\text{annual}}$ denote the annual risk-free interest rate and $r_f$ denote the per-period risk-free interest rate over $\Delta t$.

We now formalize our main assumption---that the underlying returns are skew-elliptical $t$-distributed---from which we will derive our main results of optimal option portfolios. We do so using the characterization of the skew-elliptical $t$-distribution by \citet{AzzaliniCapitanio2003Skewt}.

\begin{assumption}[Skew $t$-Distributed Returns]
    \label{assumption:skew_t_returns}
    Fix a location vector $\bs{\mu} \in \R^N$, a scale matrix $\bs{\Sigma} \in\R^{N\times N}$, a degrees of freedom $\nu \in \R$, and skewness vector $\bs{\omega} \in \R^N$. We assume that $\Delta \vec{S} $ has the multivariable skew-elliptical $t$-distribution with parameters $\big(\bs{\mu}, \bs{\Sigma}, \nu, \bs{\omega}\big)$. That is, 
    \[
        \Delta \vec{S}  \sim t_N^{\operatorname{skew}} 
        \big (
            \bs{\mu}, \bs{\Sigma}, \nu, \bs{\omega}
        \big )
        .
    \]
\end{assumption}

\subsection{Pricing the Option}

\noindent We perform option pricing under the distributional assumptions in Assumption~\ref{assumption:skew_t_returns} in the same manner as \citet{sung2026optimaloptionportfoliosskewelliptical}.

\subsection{The Delta-Gamma Approximation}

\noindent The delta-gamma approximation is the second-order Taylor expansion of the portfolio P\&L:
\[
    \Delta V (\vec{x})
    =
    (\Delta t) \theta + \bs{\delta}^\top (\Delta \vec{S}) + \dfrac{1}{2} (\Delta \vec{S})^\top \Gamma (\Delta \vec{S}).
\]
\citet{sung2026optimaloptionportfoliosskewelliptical} show that under Assumption~\ref{assumption:skew_t_returns}, the expectation and variance of the portfolio change are:
\begin{equation}
    \label{eqn:expectation_variance}
    \begin{split}
    \E[\Delta V (\vec{x}) ] &= 
         \vec{u}^\top \vec{x}
    \\
    \Var[\Delta V (\vec{x}) ] &= 
        \frac{1}{2}\vec{x}^\top Q \vec{x}
    .
    \end{split}
\end{equation}
For completeness, we include explicit forms in \ref{sec:moments}.

\subsection{Sharpe Ratio and R-VaR Ratios}

The Sharpe ratio is a risk-reward ratio which measures the excess market return achieved for each additional unit of volatility. The Sharpe ratio is given by 
\[
    \Sharpe [ \Delta V (\vec{x}) ]
    =
    \dfrac
    {
        \E[\Delta V(\vec{x})] - r_f
    }
    {
        \sqrt{
            \Var [\Delta V(\vec{x})]
        }
    }
    .
\]
In particular, the Sharpe ratio of the portfolio can be written as
\begin{equation}
    \Sharpe [ \Delta V (\vec{x}) ]
    =
    \dfrac{\vec{u}^\top \vec{x} - r_f}{\sqrt{\frac{1}{2} \vec{x}^\top Q\vec{x}}}
    .
\end{equation}
Recall the Cornish-Fisher expansion of the VaR \citep{CUI20132124,sung2026optimaloptionportfoliosskewelliptical}, defined for a tail risk threshold $\alpha >0$ as:
\[
    \CFVaR_2^\alpha [\Delta V(\vec{x})]
    =
    - \vec{u}^\top \vec{x} - \normal^{-1} (\alpha) \sqrt{\frac{1}{2}\vec{x}^\top Q\vec{x}},
\]
where $\normal^{-1}(\,\cdot\,)$ denotes the inverse c.d.f.\ of the standard normal distribution.

\begin{definition}[Return-to-VaR Ratio]
    The \textit{return-to-value-at-risk ratio}, denoted $\AdjSharpe$, is given by:
    \begin{equation}
        \AdjSharpe [\Delta V(\vec{x})] = \dfrac{\E[\Delta V(\vec{x})] - r_f}{\CFVaR_2^\alpha [\Delta V (\vec{x})]}.
    \end{equation}
    This can be written specifically for the portfolio as
    \begin{equation}
        \AdjSharpe [\Delta V(\vec{x})] = \dfrac{\vec{u}^\top \vec{x} - r_f}{-\vec{u}^\top \vec{x} - \normal^{-1}(\alpha) \sqrt{\frac{1}{2}\vec{x}^\top Q\vec{x}}}
        .
    \end{equation}
\end{definition}

\section{Portfolio Optimization} \label{sec:portfolio_optimization}

\subsection{Maximizing Sharpe Ratio}

\noindent The maximum Sharpe ratio problem is to find a portfolio allocation $\vec{x}$ which solves the fractional program
\begin{equation}
    \label{eqn:sharpe_ratio_maximization}
    \tag{P1}
    \begin{cases}
        \max \limits_{\vec{x} \in \R^M} \big\{ \Sharpe [\Delta V (\vec{x})] \big\}
        \\[10pt]
        \vec x^\top \vec v = 1.
    \end{cases}
\end{equation}

The following theorem is the first main result of our paper, and prescribes the optimal number of shares for Sharpe ratio maximization in the skew $t$-distributed returns setting.

\begin{theorem}[Optimal Sharpe Portfolios] 
    \label{thm:optimal_sharpe_portfolio}
    The Sharpe ratio maximization problem~\eqref{eqn:sharpe_ratio_maximization} is solved with the optimal portfolio allocation:
    \begin{equation}
        \label{eqn:optimal_sharpe_ratio_portfolio}
        \vec{x}_{\Sharpe}^\star = \dfrac{Q^{-1}(\vec{u} - r_f \vec{v})}{\vec{v}^\top Q^{-1} (\vec{u} - r_f \vec{v})}
        ,
    \end{equation}
    where $\vec{v}$ was defined in Equation~\eqref{eqn:option_prices_vector} and where $\vec{u}$ and $Q$ were defined below Equation~\eqref{eqn:expectation_variance}.
\end{theorem}

\subsection{Maximizing R-VaR Ratio}

The maximum R-VaR Ratio problem is to find a portfolio allocation $\vec{x}$ which solves the fractional program
\begin{equation}
    \label{eqn:var_adj_sharpe_ratio_maximization}
    \tag{P2}
    \begin{cases}
        \max \limits_{\vec{x} \in \R^M} \big\{ \AdjSharpe [\Delta V (\vec{x})] \big\}
        \\[10pt]
        \vec x^\top \vec v = 1.
    \end{cases}
\end{equation}

The second main result of our paper prescribes the optimal number of shares for R-VaR ratio maximization in the skew $t$-distributed returns settings.

\begin{theorem}[Optimal R-VaR Portfolios] 
    \label{thm:optimal_var_adj_sharpe_portfolio}
    The $\AdjSharpe$ ratio maximization problem~\eqref{eqn:var_adj_sharpe_ratio_maximization} is solved with the optimal portfolio allocation:
    \begin{equation}
        \label{eqn:optimal_rvar_ratio_portfolio}
        \vec{x}_{\AdjSharpe}^\star
        =
        \dfrac{Q^{-1} [ (1 - \lambda^\star) \vec{u} - r_f \vec{v}]}{\vec{v}^\top Q^{-1} [(1 - \lambda^\star) \vec{u} - r_f \vec{v}]}
        ,
    \end{equation}
    where $\lambda^\star$ is a constant defined in \ref{proof:thm:optimal_var_adj_sharpe_portfolio}.
\end{theorem}

\section{Distributional Parameters} \label{sec:params}

\noindent We consider option portfolios written on five underlying stocks: DIS, XOM, PFE, MO, and INTC. We consider a set of parameters produced by \citet{sung2026optimaloptionportfoliosskewelliptical} given in Table~\ref{tab:skew_t_fit}.
\begin{table}[h!]
    
    \centering
    \begin{tabular}{ll}
    \hline
    \textbf{Quantity} & \textbf{Value} \\
    \hline
    $\mathrm{dp.mu}$ & $-0.13113, -0.095028, -0.073402, 0.32902, -0.14546$ \\
    $\mathrm{dp.omega}$ & $0.069222, 0.14039, 0.14851, -0.72981, 0.17167$ \\
    $\mathrm{dp.nu}$ & $6.041$ \\
    $\mathrm{dp.\Sigma}$ & $\begin{pmatrix}0.73931 & 0.3017 & 0.2904 & 0.098984 & 0.34011\\0.3017 & 0.77136 & 0.31795 & 0.15848 & 0.25785\\0.2904 & 0.31795 & 0.71838 & 0.12605 & 0.2633\\0.098984 & 0.15848 & 0.12605 & 0.62268 & 0.076389\\0.34011 & 0.25785 & 0.2633 & 0.076389 & 0.66185\end{pmatrix}$ \\
    $\mathrm{cp.mean}$ & $-0.029077, -0.013345, 0.028808, -0.041695, -0.010414$ \\
    $\mathrm{cp.var.cov}$ & $\begin{pmatrix}1.0948 & 0.44269 & 0.42369 & 0.18581 & 0.49466\\0.44269 & 1.1465 & 0.46697 & 0.26721 & 0.37444\\0.42369 & 0.46697 & 1.0635 & 0.22632 & 0.37981\\0.18581 & 0.26721 & 0.22632 & 0.79343 & 0.16426\\0.49466 & 0.37444 & 0.37981 & 0.16426 & 0.97118\end{pmatrix}$ \\
    $\mathrm{cp.gamma1}$ & $0.096803, 0.075538, 0.098388, -0.45578, 0.13681$ \\
    $\mathrm{cp.gamma2M}$ & $35.594$ \\
    $\log L$ & $-4808.6$ \\
    \hline
    \end{tabular}
    \caption{Skew-$t$ model fit.}
    \label{tab:skew_t_fit}
\end{table}
The location vector and scale matrix for the skew $t$-distribution in this setting are:
\[
    \begin{split}
    \bs{\mu} &= \vec{S} \odot (\Delta t) \mathbf{r}
    \\
    \bs{\Sigma} &= (\Delta t)\diag (\vec{S}) \diag (\bs{\Sigma}) \vec{C} \diag (\bs{\Sigma}) \diag (\vec{S}),
    \end{split}
\]
where $\mathbf{r}$ denotes the expected annual log returns vector and $\vec{C}$ denotes the matrix of correlations.

\section{Numerical Experiments} \label{sec:experiments}

\noindent We now obtain optimal Sharpe ratio and R-VaR ratio portfolios consisting of five at-the-money options written on the five stocks in the dataset of Section~\ref{sec:params}.

\subsection{Analytic Optimization}

\noindent First, we find optimal Sharpe and R-VaR ratios based on our analytic results (Theorems~\ref{thm:optimal_sharpe_portfolio} and~\ref{thm:optimal_var_adj_sharpe_portfolio}) which are reported below.

\begin{figure}[h!]
    \centering
    \includegraphics[width=0.75\linewidth]{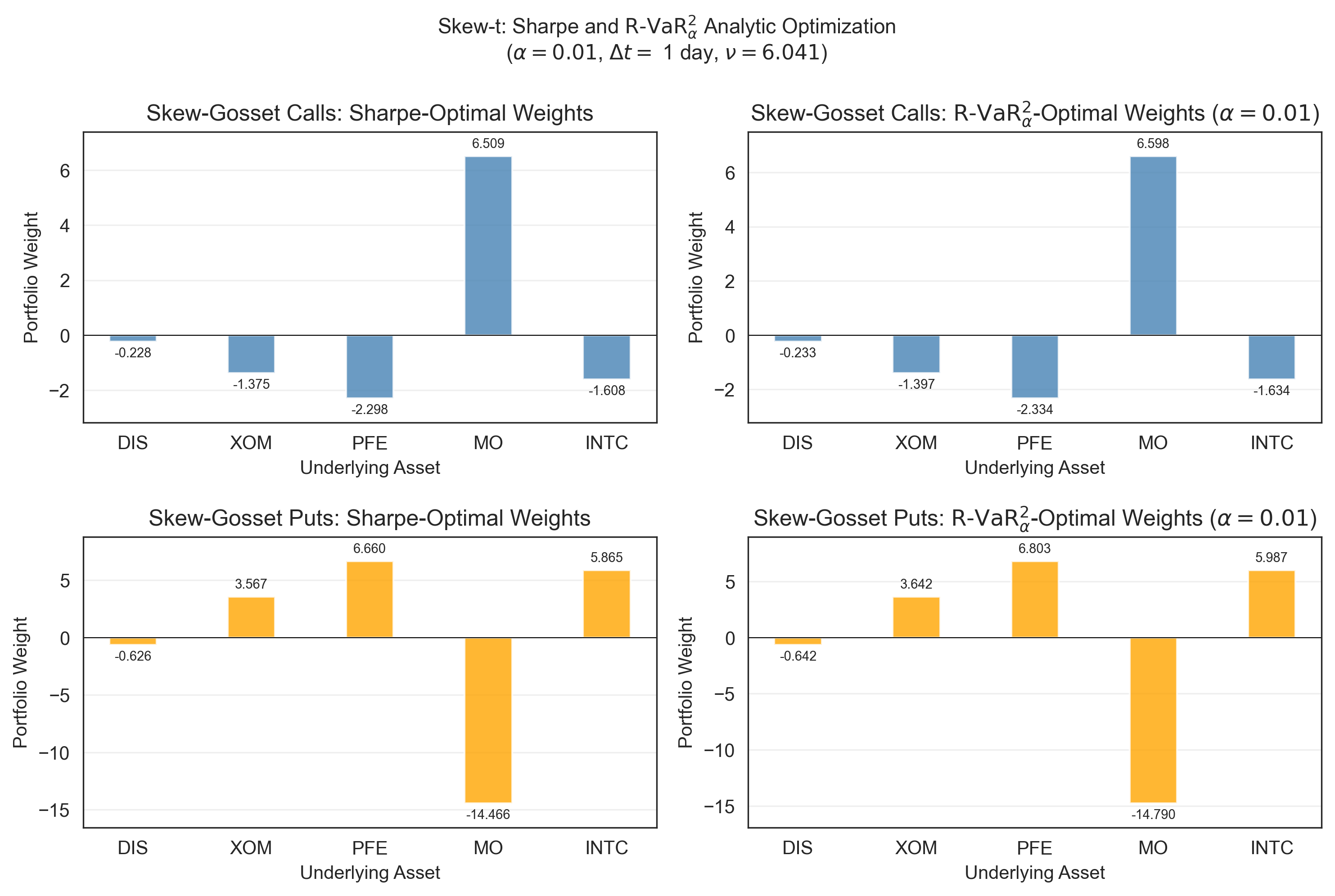}
    \caption{}
    \label{fig:analytic_portfolio}
\end{figure}

We noticed that the optimal portfolios are different between the two ratios. In particular, for the call option portfolio, there is a 9\% difference in MO stocks (see the upper-right panel in Figure~\ref{fig:analytic_portfolio}).

Optimizing the Sharpe ratio and R-VaR ratio directly appears to yield highly leveraged positions. Thus, portfolio managers may be interested to impose box constraints in order to keep exposures to various positions limited in magnitude. We do so in the following section:

\subsubsection{Effect of Box Constraints}

We now perform numerical portfolio optimization of the Sharpe ratio and R-VaR ratio with box constraints.

\begin{figure}[h!]
    \centering
    \includegraphics[width=0.75\linewidth]{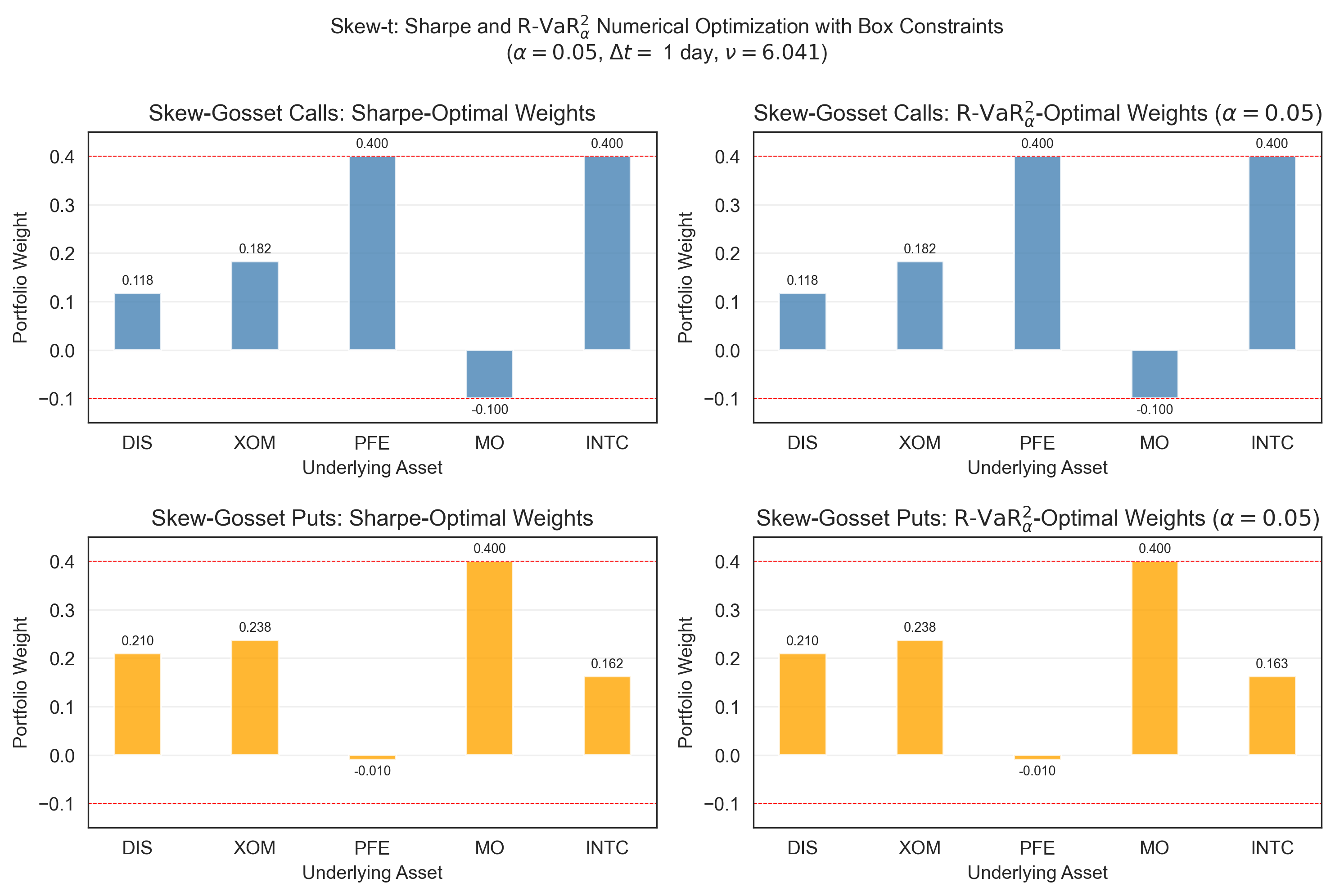}
    \caption{}
    \label{fig:numeric_box_constraints}
\end{figure}

In Figure~\ref{fig:numeric_box_constraints}, we see that imposing box constraints yields optimal positions which are very similar between Sharpe and R-VaR optimization. As such, these findings motivate a future study with  higher order approximation in the VaR expansion.

\section{Conclusion and Future Work} \label{sec:conclusion}

In this paper, we found analytical results for optimal Sharpe ratio and R-VaR ratio portfolios for options written on stocks assuming their underlying returns follow a skew $t$-distribution.

In future work, we aim to extend our results to include trading costs \citep{AlexanderS.2006MCaV}, more periods~\citep{Deng2019}, the effect of liquidity on asset prices~\citep{Shidfara2014,Yazdanian2016}, or allowing model parameters to depend on several regimes as in \citep{CoveiCanepaPirvu2019}. Another research direction is to incorporate portfolio diversification \citep{Pourbabaee2016,Maheshwari2020}.

\paragraph{Acknowledgements} Traian A. Pirvu acknowledges that this work was supported by NSERC grant RGPIN-2019-05397.

\bibliographystyle{plainnat}
\bibliography{references,ref_temp}

@misc{sung2026optimaloptionportfoliosskewelliptical,
      title={Optimal Option Portfolios for Skew-Elliptical t Returns}, 
      author={Kyle Sung and Traian A. Pirvu},
      year={2026},
      eprint={2601.07991},
      archivePrefix={arXiv},
      primaryClass={q-fin.PM},
      url={https://arxiv.org/abs/2601.07991}, 
}

@article{CUI20132124,
title = {Nonlinear portfolio selection using approximate parametric Value-at-Risk},
journal = {Journal of Banking \& Finance},
volume = {37},
number = {6},
pages = {2124-2139},
year = {2013},
issn = {0378-4266},
doi = {https://doi.org/10.1016/j.jbankfin.2013.01.036},
url = {https://www.sciencedirect.com/science/article/pii/S0378426613000617},
author = {Xueting Cui and Shushang Zhu and Xiaoling Sun and Duan Li}
}

@article{Hu01012010,
author = {Wenbo Hu and Alec N. Kercheval},
title = {Portfolio optimization for {Student} $t$ and skewed $t$ returns},
journal = {Quantitative Finance},
volume = {10},
number = {1},
pages = {91--105},
year = {2010},
publisher = {Routledge},
doi = {10.1080/14697680902814225},


URL = { 
    
        https://doi.org/10.1080/14697680902814225
    
    

},
eprint = { 
    
        https://doi.org/10.1080/14697680902814225
    
    

}

}

@book{Glasserman2003,
    author = {Paul Glasserman},
    title = {Monte Carlo Methods in Financial Engineering},
    publisher = {Springer New York, NY},
    year = {2003}
}

@article{Cassidy_2010,
   title={Pricing European options with a log Student’s t-distribution: A Gosset formula},
   volume={389},
   ISSN={0378-4371},
   url={http://dx.doi.org/10.1016/j.physa.2010.08.037},
   DOI={10.1016/j.physa.2010.08.037},
   number={24},
   journal={Physica A: Statistical Mechanics and its Applications},
   publisher={Elsevier BV},
   author={Cassidy, Daniel T. and Hamp, Michael J. and Ouyed, Rachid},
   year={2010},
   month=dec, pages={5736–5748} 
}

@article{10.1093/rfs/hhm075,
    author = {DeMiguel, Victor and Garlappi, Lorenzo and Uppal, Raman},
    title = {Optimal Versus Naive Diversification: How Inefficient is the 1/N Portfolio Strategy?},
    journal = {The Review of Financial Studies},
    volume = {22},
    number = {5},
    pages = {1915-1953},
    year = {2009},
    month = {05},
    issn = {0893-9454},
    doi = {10.1093/rfs/hhm075},
    url = {https://doi.org/10.1093/rfs/hhm075},
    eprint = {https://academic.oup.com/rfs/article-pdf/22/5/1915/24429471/hhm075.pdf},
}

@article{Wang2024SkewT,
title = {Multivariate unified skew-t distributions and their properties},
journal = {Journal of Multivariate Analysis},
volume = {203},
pages = {105322},
year = {2024},
issn = {0047-259X},
doi = {https://doi.org/10.1016/j.jmva.2024.105322},
url = {https://www.sciencedirect.com/science/article/pii/S0047259X24000290},
author = {Kesen Wang and Maicon J. Karling and Reinaldo B. Arellano-Valle and Marc G. Genton}
}

@article{Cassidy01082013,
    author = {Daniel T. Cassidy and Michael J. Hamp and Rachid Ouyed},
    title = {Log Student’s t-distribution-based option sensitivities: Greeks for the Gosset formulae},
    journal = {Quantitative Finance},
    volume = {13},
    number = {8},
    pages = {1289--1302},
    year = {2013},
    publisher = {Routledge},
    doi = {10.1080/14697688.2012.744087},
    URL = { 
            https://doi.org/10.1080/14697688.2012.744087
    },
    eprint = { 
            https://doi.org/10.1080/14697688.2012.744087
    }
}

@article{basnarkov2019optionpricingheavytaileddistributions,
      title={Option Pricing with Heavy-Tailed Distributions of Logarithmic Returns}, 
      author={Lasko Basnarkov and Viktor Stojkoski and Zoran Utkovski and Ljupco Kocarev},
      year={2019},
      journal={International Journal of Theoretical and 
Applied Finance},
        volume={22},
        number={7},
      eprint={1807.01756},
    pages={1950041},
}

@article{Maheshwari2020,
  author    = {A. Maheshwari and T. A. Pirvu},
  title     = {Portfolio Optimization under Correlation Constraint},
  journal   = {Risks},
  year      = {2020},
  volume    = {8},
  number    = {1},
  pages     = {1--19},
}

@article{CoveiCanepaPirvu2019,
  author  = {Covei, D. P. and Canepa, E. C. and Pirvu, T. A.},
  title   = {Stochastic production planning with regime switching},
  journal = {Journal of Industrial and Management Optimization},
  volume  = {19},
  number  = {3},
  pages   = {1697--1713},
  year    = {2019}
}

@article{Pourbabaee2016,
  author    = {F. Pourbabaee and M. Kwak and T. A. Pirvu},
  title     = {Risk Minimization and Portfolio Diversification},
  journal   = {Quantitative Finance},
  year      = {2016},
  volume    = {9},
  number    = {1},
  pages     = {1325--1332}
}

@article{Yazdanian2016,
  author    = {A. R. Yazdanian and T. A. Pirvu},
  title     = {Numerical Analysis for Spread Option Pricing Model of Markets with Finite Liquidity: Full Feedback Model},
  journal   = {Applied Mathematics \& Information Sciences},
  year      = {2016},
  volume    = {10},
  number    = {4},
  pages     = {1271--1281}
}

@article{Deng2019,
  author    = {L. Deng and T. A. Pirvu},
  title     = {Multi-period Investment Strategies under Cumulative Prospect Theory},
  journal   = {Journal of Risk and Financial Management},
  year      = {2019},
  volume    = {12},
  number    = {2},
  pages     = {1--15}
}

@article{Shidfara2014,
  author    = {A. Shidfara and K. Paryaba and A. R. Yazdanian and T. A. Pirvu},
  title     = {Numerical Analysis for Spread Option Pricing Model of Markets with Finite Liquidity: First-order Feedback Model},
  journal   = {International Journal of Computer Mathematics},
  year      = {2014},
  volume    = {91},
  number    = {12},
  pages     = {2603--2620}
}

@article{Jewell2013,
  author    = {S. Jewell and Y. Li and T. A. Pirvu},
  title     = {Non-Linear Equity Portfolio Variance Reduction under Delta-Gamma Approach Analysis},
  journal   = {Operations Research Letters},
  year      = {2013},
  volume    = {41},
  number    = {6},
  pages     = {694--700}
}

@article{Chen2013,
  author    = {R. Chen and L. Yu},
  title     = {A Novel Nonlinear Value-at-Risk Method for Modeling Risk of Option Portfolio with Multivariate Mixture of Normal Distributions},
  journal   = {Economic Modelling},
  year      = {2013},
  volume    = {35},
  pages     = {796--804}
}

@article{ZAKAMOULINE20091242,
title = {Portfolio performance evaluation with generalized Sharpe ratios: Beyond the mean and variance},
journal = {Journal of Banking \& Finance},
volume = {33},
number = {7},
pages = {1242-1254},
year = {2009},
issn = {0378-4266},
doi = {https://doi.org/10.1016/j.jbankfin.2009.01.005},
url = {https://www.sciencedirect.com/science/article/pii/S0378426609000132},
author = {Valeri Zakamouline and Steen Koekebakker}
}

@article{maller2010,
    author={Maller, Ross A. and Durand, Robert B. and Jafarpour, Hediah},
    year={2010},
    month={Summer},
    title={Optimal portfolio choice using the maximum Sharpe ratio},
    journal={The Journal of Risk},
    volume={12},
    number={4},
    pages={49-73},
    note={Copyright - Copyright Incisive Media Plc Summer 2010; Document feature - Equations; Graphs; ; Last updated - 2024-11-19; SubjectsTermNotLitGenreText - United States--US},
    isbn={14651211},
    language={English},
    url={http://libaccess.mcmaster.ca/login?url=https://www.proquest.com/scholarly-journals/optimal-portfolio-choice-using-maximum-sharpe/docview/817784573/se-2},
}

@article{sharpe1998sharpe,
  title={The {Sharpe} ratio},
  author={Sharpe, William F},
  journal={Streetwise--the Best of the Journal of Portfolio Management},
  volume={3},
  number={3},
  pages={169--85},
  year={1998},
  publisher={Princeton University Press NJ}
}

@article{Cogneau2009_101Ways,
author = {Cogneau, Philippe and Hübner, Georges},
year = {2009},
month = {01},
pages = {},
title = {The 101 Ways to Measure Portfolio Performance},
journal = {SSRN Electronic Journal},
doi = {10.2139/ssrn.1326076}
}

@article{sortino1991,
author={Sortino,Frank A. and van der Meer,Robert},
year={1991},
month={Summer},
title={Downside Risk},
journal={Journal of Portfolio Management},
volume={17},
number={4},
pages={27},
note={Copyright - Copyright Euromoney Institutional Investor PLC Summer 1991; Last updated - 2024-12-01},
isbn={00954918},
language={English},
url={http://libaccess.mcmaster.ca/login?url=https://www.proquest.com/scholarly-journals/downside-risk/docview/195576478/se-2},
}

@article{MartinR.Douglas2003Popa,
author = {Martin, R. Douglas and Rachev, Svetlozar (Zari) and Siboulet, Frederic},
issn = {1540-6962},
journal = {Wilmott (London, England)},
language = {eng},
number = {6},
pages = {70-83},
title = {Phi-alpha optimal portfolios and extreme risk management},
volume = {2003},
year = {2003},
}

@article{Cotter01102006,
author = {John Cotter},
title = {Extreme Value Estimation of Boom and Crash Statistics},
journal = {The European Journal of Finance},
volume = {12},
number = {6-7},
pages = {553--566},
year = {2006},
publisher = {Routledge},
doi = {10.1080/13518470500460111},
URL = { https://doi.org/10.1080/13518470500460111 },
eprint = {
https://doi.org/10.1080/13518470500460111
}
}

@article{Echaust20142234,
author = {Echaust, Krzysztof},
doi = {10.2478/foli-2014-0102},
url = {https://doi.org/10.2478/foli-2014-0102},
title = {A Comparison of Tail Behaviour of Stock Market Returns},
journal = {Folia Oeconomica Stetinensia},
number = {1},
volume = {14},
year = {2014},
pages = {22--34}
}

@article{KIM2003417,
    title = {Moments of random vectors with skew $t$ distribution and their quadratic forms},
    journal = {Statistics \& Probability Letters},
    volume = {63},
    number = {4},
    pages = {417-423},
    year = {2003},
    issn = {0167-7152},
    doi = {https://doi.org/10.1016/S0167-7152(03)00121-4},
    url = {https://www.sciencedirect.com/science/article/pii/S0167715203001214},
    author = {Hyoung-Moon Kim and Bani K. Mallick}
}

@article{AzzaliniCapitanio2003Skewt,
    author = {Azzalini, Adelchi and Capitanio, Antonella},
    title = {Distributions Generated by Perturbation of Symmetry with Emphasis on a Multivariate Skew $t$-Distribution},
    journal = {Journal of the Royal Statistical Society Series B: Statistical Methodology},
    volume = {65},
    number = {2},
    pages = {367-389},
    year = {2003},
    month = {04},
    issn = {1369-7412},
    doi = {10.1111/1467-9868.00391},
    url = {https://doi.org/10.1111/1467-9868.00391},
    eprint = {https://academic.oup.com/jrsssb/article-pdf/65/2/367/49684466/jrsssb_65_2_367.pdf},
}

@Article{YinBalakrishnan2024SkewElliptical,
journal={Journal of Multivariate Analysis},
author={Yin, Chuancun and Balakrishnan, Narayanaswamy},
title={Stochastic representations and probabilistic characteristics of multivariate skew-elliptical distributions},
year={2024},
volume={199},
number={C},
doi={10.1016/j.jmva.2023.105240},
url={https://ideas.repec.org/a/eee/jmvana/v199y2024ics0047259x23000866.html},
}

@article{AzzaliniCapitanio1999SkewNormal,
 ISSN = {13697412, 14679868},
 URL = {http://www.jstor.org/stable/2680724},
 author = {A. Azzalini and A. Capitanio},
 journal = {Journal of the Royal Statistical Society. Series B (Statistical Methodology)},
 number = {3},
 pages = {579--602},
 publisher = {[Royal Statistical Society, Oxford University Press]},
 title = {Statistical Applications of the Multivariate Skew Normal Distribution},
 urldate = {2026-03-13},
 volume = {61},
 year = {1999}
}

@article{genton2005generalized,
  title   = {Generalized skew-elliptical distributions and their quadratic forms},
  author  = {Genton, Marc G. and Loperfido, Nicola M. R.},
  journal = {Annals of the Institute of Statistical Mathematics},
  volume  = {57},
  number  = {2},
  pages   = {389--401},
  year    = {2005},
  doi     = {10.1007/BF02507031}
}

@Article{risks5020027,
AUTHOR = {Metel, Michael R. and A. Pirvu, Traian and Wong, Julian},
TITLE = {Risk Management under Omega Measure},
JOURNAL = {Risks},
VOLUME = {5},
YEAR = {2017},
NUMBER = {2},
ARTICLE-NUMBER = {27},
URL = {https://www.mdpi.com/2227-9091/5/2/27},
ISSN = {2227-9091},
DOI = {10.3390/risks5020027}
}

@article{Lin2007RobustSkewT,
  author  = {Lin, T. I. and Lee, J. C. and Hsieh, W. J.},
  title   = {Robust mixture modeling using the skew $t$ distribution},
  journal = {Statistics and Computing},
  year    = {2007},
  volume  = {17},
  number  = {1},
  pages   = {81--92},
  doi     = {10.1007/s11222-006-9005-8},
  publisher = {Springer}
}

@article{GENTON2001319,
title = {Moments of skew-normal random vectors and their quadratic forms},
journal = {Statistics \& Probability Letters},
volume = {51},
number = {4},
pages = {319-325},
year = {2001},
issn = {0167-7152},
doi = {https://doi.org/10.1016/S0167-7152(00)00164-4},
url = {https://www.sciencedirect.com/science/article/pii/S0167715200001644},
author = {Marc G. Genton and Li He and Xiangwei Liu}
}

@article{BRANCO200199,
title = {A General Class of Multivariate Skew-Elliptical Distributions},
journal = {Journal of Multivariate Analysis},
volume = {79},
number = {1},
pages = {99-113},
year = {2001},
issn = {0047-259X},
doi = {https://doi.org/10.1006/jmva.2000.1960},
url = {https://www.sciencedirect.com/science/article/pii/S0047259X00919602},
author = {Márcia D. Branco and Dipak K. Dey}
}

@article{AlexanderS.2006MCaV,
pages = {583-605},
publisher = {Elsevier B.V},
author = {Alexander, S. and Coleman, T.F. and Li, Y.},
address = {Amsterdam},
copyright = {2005},
issn = {0378-4266},
journal = {Journal of Banking \& Finance},
language = {eng},
number = {2},
volume = {30},
year = {2006},
title = {Minimizing CVaR and VaR for a portfolio of derivatives},
}

@article{AlexanderGordonJ.2003PPEU,
author = {Alexander, Gordon J. and Baptista, Alexandre M.},
address = {London},
copyright = {Copyright Euromoney Institutional Investor PLC Summer 2003},
issn = {0095-4918},
journal = {Journal of portfolio management},
language = {eng},
number = {4},
pages = {93-102},
publisher = {Pageant Media},
title = {Portfolio Performance Evaluation Using Value at Risk},
volume = {29},
year = {2003},
}

@article{GregoriouGregN2004PoCh,
author = {Gregoriou, Greg N},
address = {London},
copyright = {Copyright Henry Stewart Conferences and Publications Ltd. 2004},
issn = {1753-9641},
journal = {Journal of derivatives \& hedge funds},
language = {eng},
number = {2},
pages = {149-},
publisher = {Palgrave Macmillan},
title = {Performance of {C}anadian hedge funds using a modified {S}harpe ratio},
volume = {10},
year = {2004},
}

\appendix 

\section{Moment Explicit Forms} \label{sec:moments}
\noindent In Equation~\eqref{eqn:expectation_variance}, $\vec{u}$ and $Q$ can be written as:
\begin{align*}
    \mbf u &\coloneqq \bs{\zeta} + cB \mbf h + c D \mbf h,
    \\
    Q &\coloneqq \frac{1}{2} (\tilde Q + \tilde Q^\top),
\end{align*}
where
    \begin{align*}
        c &\coloneqq \sqrt{\dfrac{\nu}{\pi}} \dfrac{\operatorname{Gamma}(\tfrac{\nu - 1}{2})}{\operatorname{Gamma}(\tfrac{\nu}{2})}
        \\
        \vec{h} &\coloneqq \dfrac{\bs{\Sigma} \bs{\omega}}{\sqrt{1 + \bs{\omega}^\top \bs{\Sigma} \bs{\omega}}}
    \end{align*}
are skew distributional parameters,
and the intermediate variables are:
\begin{align*}
    \tilde{Q} &\coloneqq U + \dfrac{4c\nu}{\nu - 3}(H + E) + \dfrac{2c\nu}{(\nu - 2)(\nu - 3)} (B + D^\top)\vec{h}\vec{p}^\top 
    \\&\quad\; 
    - \dfrac{2c\nu}{\nu - 3}(B + D^\top)\vec{h}\mathbf{q}^\top - 2c^2 (B + D^\top)\vec{h}\vec{h}^\top(B + D^\top)^\top
    \\
    \bs{\zeta} &\coloneqq 
        \left[ (\Delta t) \bs{\theta} + D^\top \bs{\mu} + \dfrac{\nu}{2(\nu - 2)} \vec{p} + \bs{\xi}  \right]
    \\
    U &\coloneqq 
        \left[\dfrac{2\nu}{\nu -2} ((D^\top + B)\bs{\Sigma} (D^\top + B)^\top) + \dfrac{\nu^2}{(\nu - 2) (\nu -4)} R + \dfrac{\nu^2}{(\nu-2)^2(\nu-4)}\vec{p}\vec{p}^\top\right]
\end{align*}
\[
    \begin{alignedat}{2}
        H
        &\coloneqq [h_{i,j}]_{i,j=1}^N, &\quad
         h_{i,j} &\coloneqq \bs{\mu}^\top \Gamma^{[i]} \bs{\Sigma} \Gamma^{[j]} \vec{h},
        \\
        E &\coloneqq [e_{i,j}]_{i,j=1}^N, &\quad e_{i,j} &\coloneqq (D_{[\bullet, i]})^\top \bs{\Sigma} \Gamma^{[j]} \vec{h},
        \\
        \mathbf{q} &\coloneqq [q_n]_{n=1}^N, \quad &q_n &\coloneqq \vec{h}^\top \Gamma^{[n]} \vec{h}
        ,
    \end{alignedat}
\]
\[
    \begin{aligned}
        \vec{p} &\coloneqq (p_1,\ldots,p_M)  &  \quad p_m &\coloneqq \tr[\Gamma^{[m]} \bs{\Sigma} ]
        \\
        D &\coloneqq (\delta_n^m)  &  \quad \delta_n^m &\coloneqq \pdv{V_m}{S_n}
        \\
        R &\coloneqq [r_{i,j}] & \quad r_{i,j} &\coloneqq \tr[\Gamma^{[i]}\bs{\Sigma} \Gamma^{[j]} \bs{\Sigma}]
        \\
        \bs{\xi} &\coloneqq (\xi_1,\ldots,\xi_M)  & \quad \xi_m &\coloneqq \dfrac{1}{2} \sum\limits_{i=1}^N \sum\limits_{j=1}^N \mu_i \mu_j {\gamma}^{[i,j]}_m 
        \\
                B &\coloneqq \begin{bmatrix}
                \text{---} \hspace{-0.025cm} & \bs{\mu}^\top \Gamma^{[1]} & \hspace{-0.025cm} \text{---}
                \\
                \text{---} \hspace{-0.025cm} & \vdots& \hspace{-0.025cm} \text{---}
                \\
                \text{---} \hspace{-0.025cm} & \bs{\mu}^\top \Gamma^{[M]} & \hspace{-0.025cm} \text{---}
            \end{bmatrix}
        .
    \end{aligned}
\]

\section{Proofs} \label{sec:proofs}

\subsection{Proof of Theorem~\ref{thm:optimal_sharpe_portfolio}}

\begin{proof}
    We proceed with a modification to Lagrange multipliers by setting $\sqrt{\frac{1}{2} \vec{x}^\top Q\vec{x}} = \varepsilon$, maximizing the linear term, then maximizing over all possible $\varepsilon$. Thus, we have the modified problem: 
    \begin{equation*}
        \label{eqn:proof:thm:optimal_sharpe_portfolios:approach2:modified_problem}
        \begin{cases}
        \max \limits_{\vec{x} \in \R^M} \dfrac{\vec{u}^\top \vec{x} - r_f}{\varepsilon} \equiv \max \limits_{\vec{x} \in \R^M} \vec{u}^\top \vec{x}
            \\[10pt]
        \vec{x}^\top \vec{v} = 1
        \\[10pt]
        \frac{1}{2} \vec{x}^\top Q\vec{x} = \varepsilon^2
        .
        \end{cases}
    \end{equation*}
    After solving the problem in Equation~\eqref{eqn:proof:thm:optimal_sharpe_portfolios:approach2:modified_problem} and achieving $\vec{x}^\star(\varepsilon)$, we can find the optimal $\varepsilon$ by taking the maximum over $\varepsilon \in [0,\infty)$ of $\dfrac{\vec{u}^\top \vec{x}^\star(\varepsilon) - r_f}{\varepsilon}$.

    \noindent We now construct the Lagrangian:
    \begin{equation*}
        \label{eqn:proof:thm:optimal_sharpe_portfolios:approach2:lagrangian}
        \Lc ( \vec{x}; \lambda_1, \lambda_2)
        =
        \vec{u}^\top \vec{x} + \lambda_1 (\vec{x}^\top \vec{v} - 1) + \lambda_2 ((\tfrac{1}{2}\vec{x}^\top Q\vec{x})-\varepsilon^2).
    \end{equation*}
    By the K.K.T.\ stationarity condition, we have
    \[
        \nabla_\vec{x} \Lc ( \vec{x}; \lambda_1, \lambda_2 )
        =
        \vec{u} + \lambda_1 \vec{v} + \lambda_2 Q\vec{x}
        =
        \Zero
        .
    \]
    Isolating, we achieve
    \begin{equation}
        \label{eqn:proof:thm:optimal_sharpe_portfolios:approach2:x_lagrangian}
        \vec{x} 
        = 
            - \frac{1}{\lambda_2} Q^{-1}(\vec{u} + \lambda_1 \vec{v})
        .
    \end{equation}
    Left multiplying by $\vec{v}$ and applying the constraint $\vec{x}^\top \vec{v} = 1$, we have \[
        \vec{v}^\top \vec{x} = -\dfrac{1}{\lambda_2} \left( \vec{v}^\top Q^{-1}\vec{u} +\lambda_1 \vec{v}^\top Q^{-1} \vec{v}\right) = 1,
    \]
    and so defining the intermediate variables
    \begin{equation*}
        \label{eqn:proof:thm:optimal_sharpe_portfolios:subspace_maximzation:INTERMEDIATE_VARS}
            \rho_1 \coloneqq \vec{u}^\top Q^{-1} \vec{u}
            ,\qquad
            \rho_2 \coloneqq \vec{u}^\top Q^{-1} \vec{v}
            ,\qquad
            \rho_3 \coloneqq \vec{v}^\top Q^{-1} \vec{v},
    \end{equation*}
    we have
    \begin{equation}
        \label{eqn:proof:thm:optimal_sharpe_portfolios:approach2:lagrange_multipliers:2}
        \lambda_2 
        = 
        - (\rho_2 + \lambda_1 \rho_3)
        .
    \end{equation}
    Thus, substituting the first Lagrange multiplier in $\vec{x}$ yields:
    \begin{equation}
        \label{eqn:proof:thm:optimal_sharpe_portfolios:approach2:x_lagrangian:lambda2_known}
        \vec{x} = \dfrac{1}{\rho_2 + \lambda_1 \rho_3} Q^{-1} (\vec{u} + \lambda_1 \vec{v}).
    \end{equation}
    Similarly, substituting Equation~\eqref{eqn:proof:thm:optimal_sharpe_portfolios:approach2:x_lagrangian:lambda2_known} into the constraint $\frac{1}{2}\vec{x}^\top Q\vec{x} = \varepsilon^2$ yields
    \[
        \frac{1}{2} \frac{1}{(\rho_2 + \lambda_1\rho_3)^2} (\vec{u} + \lambda_1\vec{v})^\top Q^{-1} (\vec{u} + \lambda_1 \vec{v}) = \varepsilon^2
        ,
    \]
    and after some manipulation, we produce
    \[
        \lambda_1^2 (\rho_3 - 2\varepsilon^2\rho_3^2) + \lambda_1 (2\rho_2 - 4\varepsilon^2 \rho_2\rho_3) + (\rho_1 - 2\varepsilon^2 \rho_2^2) = 0.
    \]
    Solving for $\lambda_1$ produces:
    \begin{align*}
        \lambda_1 
        &= 
            \dfrac{-\rho_2(1 - 2\varepsilon^2 \rho_3) \pm \sqrt{(1 - 2\varepsilon^2 \rho_3) (\rho_2^2 - \rho_1\rho_3)}}{\rho_3 (1 - 2\varepsilon^2 \rho_3)}
        .
    \end{align*}
    Returning to $\lambda_2$, we substitute $\lambda_1$ in Equation~\eqref{eqn:proof:thm:optimal_sharpe_portfolios:approach2:lagrange_multipliers:2} to produce
    {\allowdisplaybreaks \begin{align*}
        \lambda_2
        &=
            \mp \sqrt{\dfrac{\rho_2^2 - \rho_1\rho_3}{1 - 2\varepsilon^2 \rho_3}} 
        .
    \end{align*}}
    Substituting $\lambda_1$ and $\lambda_2$ in Equation~\eqref{eqn:proof:thm:optimal_sharpe_portfolios:approach2:x_lagrangian} produces the expression 
    \begin{align}
        \vec{x} (\varepsilon,\pm)
        &= 
        \label{eqn:proof:thm:optimal_sharpe_portfolios:approach2:x_in_lagrange_multipliers}
            \pm Q^{-1}\left( \sqrt{\frac{1-2\varepsilon^2\rho_3}{\rho_2^2-\rho_1\rho_3}}\;\vec{u} + \frac{-\rho_2\sqrt{\frac{1-2\varepsilon^2\rho_3}{\rho_2^2-\rho_1\rho_3}}\pm 1}{\rho_3}\,\vec{v} \right)
        .
    \end{align}
    This now produces the modified objective
    \[
        \begin{cases}
            \max \limits_{\varepsilon > 0} \dfrac{\vec{u}^\top \vec{x}(\varepsilon, \pm) - r_f}{\varepsilon}
            \\[10pt]
            \vec{x}^\top \vec{v} = 1
            \\[5pt]
            \tfrac{1}{2} \vec{x}^\top Q\vec{x} = \varepsilon^2.
        \end{cases}
    \]
    We may substitute Equation~\eqref{eqn:proof:thm:optimal_sharpe_portfolios:approach2:x_in_lagrange_multipliers} in the objective of the modified problem in Equation~\eqref{eqn:proof:thm:optimal_sharpe_portfolios:approach2:modified_problem} to produce:
    \begin{align*}
        f (\varepsilon) 
        &\coloneqq 
            \frac{1}{\varepsilon} \left(
                \frac{\rho_2}{\rho_3} - r_f \mp \frac{1}{\rho_3} \sqrt{(\rho_2^2 - \rho_1\rho_3)(1 - 2\varepsilon^2 \rho_3)}
            \right)
        .
    \end{align*}
    We perform a univariable maximization over $f(\varepsilon)$ for each positive $\varepsilon$. First, we compute the derivative $f'(\varepsilon)$ as:
    \begin{align*}
        f'(\varepsilon)
        &=    -\frac{1}{\varepsilon^{2}}
    \left(
        \frac{\rho_{2}}{\rho_{3}}
        - r_{f}
        \mp
        \frac{1}{\rho_{3}}
        \sqrt{\left(\rho_{2}^{2}-\rho_{1}\rho_{3}\right)\left(1-2\varepsilon^{2}\rho_{3}\right)}
        \right)
        \pm
        \frac{2\left(\rho_{2}^{2}-\rho_{1}\rho_{3}\right)}
        {\sqrt{\left(\rho_{2}^{2}-\rho_{1}\rho_{3}\right)\left(1-2\varepsilon^{2}\rho_{3}\right)}}.
    \end{align*}
    We find the stationary points $f'(\varepsilon) = 0$, 
    and note that only the positive root 
    is admissible since we assume $\varepsilon > 0$. So, we have:
    \[
        \varepsilon_\star
        =
        \sqrt{
            \dfrac{\rho_1 - 2r_f \rho_2 + r_f^2 \rho_3}{2(\rho_2 - r_f \rho_3)^2}
        }
        .
    \]
    Substituting these values in $\vec{x}(\varepsilon)$ in Equation~\eqref{eqn:proof:thm:optimal_sharpe_portfolios:approach2:x_in_lagrange_multipliers}, we obtain:
    {\allowdisplaybreaks
        \begin{align*}
        \vec{x}(\varepsilon^\star, \pm) 
        &= 
            \pm Q^{-1}\left(
                \frac{1}{\lvert \rho_2-r_f\rho_3\rvert}\vec{u}
                +
                \frac{
                    -\dfrac{\rho_2}{\lvert \rho_2-r_f\rho_3\rvert}\pm 1
                }{\rho_3} \vec{v}
            \right)
            \!.
    \end{align*}
    }
    Now, for the other $\pm$, we choose the branch that maximizes $f(\varepsilon)$, that is, the sign such that \[
        \pm |\rho_2 - r_f \rho_3| = \rho_2 - r_f \rho_3.
    \]
    Some simplifications and returning our intermediate variables to the original variables yields the result in equation~\eqref{eqn:optimal_sharpe_ratio_portfolio}.
\end{proof}

\subsection{Proof of Theorem~\ref{thm:optimal_var_adj_sharpe_portfolio}} \label{proof:thm:optimal_var_adj_sharpe_portfolio}

\begin{proof}
    We proceed with a modification to Lagrange multipliers: we will solve a constrained quadratic problem and then reduce to a univariable optimization problem. We do this by setting the linear term to
    a constant $\varepsilon$. We then will maximize the remaining fractional term, then perform a univariable maximization over all admissible $\varepsilon$.

    Setting $\varepsilon \coloneqq \vec{u}^\top \vec{x}$, for fixed $\varepsilon$:
    \begin{equation}
        \notag
        \begin{cases}
            \max \limits_{\vec{x} \in \R^M} \dfrac{\varepsilon - r_f}{\varepsilon - \normal^{-1}(\alpha)\sqrt{\frac{1}{2}\vec{x}^\top U\vec{x}}}
            \\[10pt]
            \vec{x}^\top \vec{v}(\vec{S},t) = 1
            \\
            \vec{u}^\top \vec{x} = \varepsilon
            .
        \end{cases}
    \end{equation}
    By assumption that $\alpha < 1/2$, the expression $\normal^{-1}(\alpha)$ is negative, thus the denominator is strictly positive. Thus, this maximization is solved when the radical term is minimized. Furthermore, the monotonicty of the square root implies that the optimization problem is simply the minimum of the quadratic form. Thus, we achieve the following optimization problem:
    \[
        \begin{cases}
            \min \limits_{\vec{x} \in \R^M} \left\{ 
            \dfrac{1}{2} \vec{x}^\top U \vec{x}
            \right\}
            \\[12pt]
            J \vec{x} = \bs{\psi}
            ,
        \end{cases}
    \]
    where $J$ and $\bs{\psi}$ are given by
    \begin{equation}
        \label{eqn:J_psistar}
        J \coloneqq \begin{bmatrix}
            \text{---}\!\!&\!\! \mbf u^\top  \!\!&\!\! \text{---}
            \\
            \text{---}\!\!&\!\!  \vec{v}(\vec{S},t)^\top  \!\!&\!\! \text{---}
        \end{bmatrix},
        \quad
        \bs{\psi} \coloneqq \begin{bmatrix}
            \varepsilon
            \\
            1
        \end{bmatrix}.
    \end{equation}
    By Theorem 3.1 \citep{sung2026optimaloptionportfoliosskewelliptical}, the optimal solution $\vec{x}(\varepsilon)$ satisfies
    \begin{align*}
        \vec{x}(\varepsilon) &= G\bs{\psi}(\varepsilon)
        \\
        \frac{1}{2}\vec{x}^\top U\vec{x} &= \Ascr \varepsilon^2 
            + 
            \Bscr \varepsilon 
            + 
            \Cscr
            ,
    \end{align*}
    where $G \coloneqq U^{-1}J^\top (JU^{-1}J^\top)^{-1}$ and 
    \begin{align*}
        \Ascr \coloneqq \frac{1}{2} G_{[\bullet,1]}^\top U G_{[\bullet, 1]}^{}
        ,\qquad
        \Bscr \coloneqq G_{[\bullet,2]}^\top U G_{[\bullet, 1]}
        ,\qquad
        \Cscr \coloneqq \frac{1}{2} G_{[\bullet,2]}^\top U G_{[\bullet, 2]}^{}
        .
    \end{align*}
    Returning to our original objective yields
    \begin{equation}
        \notag
        \begin{cases}
            \max \limits_{\varepsilon \geq 0} \dfrac{\varepsilon - r_f}{\varepsilon - \normal^{-1}(\alpha) \sqrt{\Ascr \varepsilon^2 
            + 
            \Bscr \varepsilon 
            + 
            \Cscr}}
            \\[10pt]
            \vec{x}^\top \vec{v}(\vec{S},t) = 1
            \\
            \vec{u}^\top \vec{x} = \varepsilon
            .
        \end{cases}
    \end{equation}    
    Thus, we must find the maximizer of the objective function \[
    f(\varepsilon) \coloneqq
        \dfrac{\varepsilon - r_f}{\varepsilon - \normal^{-1}(\alpha) \sqrt{\Ascr \varepsilon^2 
            + 
            \Bscr \varepsilon 
            + 
            \Cscr}}
            .
    \]
    We do this by first order optimization. Set the following intermediate variables:
    \begin{equation}
    \label{eqn:proof:thm:optimal_RVaR:A_B_C_reformulated}
    \begin{split}
        \Acal_{\AdjSharpe} &= \left[
            4r_f^2 \Ascr - \normal^{-1}(\alpha)^2 [\Bscr + 2 r_f \Ascr]^2
        \right]
        \\
        \Bcal_{\AdjSharpe} &= \left[
            4r_f^2 \Bscr - 2 \normal^{-1}(\alpha)^2[\Bscr  + 2 r_f \Ascr ] [2\Cscr  + r_f  \Bscr ]
        \right]
        \\
        \Ccal_{\AdjSharpe} &= \left[
            4 r_f^2 \Cscr - \normal^{-1}(\alpha)^2[2\Cscr  + r_f  \Bscr ]^2
        \right]
    \!
    .
    \end{split}
    \end{equation}
    Taking derivatives, one can show that the critical points of $f$ are: 
    \begin{equation}
        \label{eqn:proof:thm:optimal_RVaR:eps_PM}
        \varepsilon_\pm = \dfrac{-\Bcal_{\AdjSharpe} \pm \sqrt{\Bcal_{\AdjSharpe}^2 - 4\Acal_{\AdjSharpe}\Ccal_{\AdjSharpe}}}{2\Acal_{\AdjSharpe}}
        .
    \end{equation}
    Thus, the optimal solution is then
    \[
        \vec{x}_\star
        =
        G\bs{\psi}(\varepsilon_\pm).
    \]
    One can then show, through algebraic manipulation, that this result can be written in the form given in equation~\eqref{eqn:optimal_rvar_ratio_portfolio}.
\end{proof}

\end{document}